\documentclass[12pt]{article}
\usepackage[latin1]{inputenc}
\usepackage[british]{babel}
\usepackage{cmap}
\usepackage{lmodern}

\usepackage{amssymb, amsmath, amsthm}
\usepackage[a4paper,top=25mm,bottom=25mm,left=25mm,right=25mm]{geometry}
\usepackage{etex}
\usepackage{ragged2e}

\usepackage{authblk} % for headings
\usepackage{pifont}
\usepackage{graphicx}
\usepackage[usenames,dvipsnames,svgnames,table]{xcolor}
\usepackage[figuresright]{rotating}
\usepackage{xtab} % tackle the long tables
\usepackage{longtable} % tackle the long tables
\usepackage{multirow}
\usepackage{footnote}
\usepackage[stable]{footmisc}
\usepackage{chngpage} % allows for temporary adjustment of side margins
\usepackage{pdflscape} % landscape environment
\usepackage[nottoc,notlot,notlof]{tocbibind} % includes everything in ToC

\usepackage{pgfplots}
\pgfplotsset{every tick label/.append style={font=\small}}
\pgfplotsset{compat=1.14}
\usepackage{setspace}

\usepackage{array}
\newcolumntype{K}[1]{>{\centering\arraybackslash$}p{#1}<{$}}

\makesavenoteenv{tabular}
\usepackage{tabularx}
\usepackage{booktabs}
\usepackage{threeparttable} % [flushleft]
\usepackage[referable]{threeparttablex} % footnotes in tabu
\newcolumntype{R}{>{\raggedleft\arraybackslash}X}
\newcolumntype{L}{>{\raggedright\arraybackslash}X}
\newcolumntype{C}{>{\centering\arraybackslash}X}
\newcolumntype{A}{>{\columncolor{gray!25}}C}
\newcolumntype{a}{>{\columncolor{gray!25}}c}

\newlength{\tablen}

\usepackage{dcolumn} % alignment to decimal points
\newcolumntype{.}{D{.}{.}{-1}}

\usepackage{tikz}
\usetikzlibrary{arrows, calc, matrix, patterns, positioning, trees}
\usepackage[semicolon]{natbib}
\usepackage[hyphens]{url}
\usepackage{hyperref} % [hidelinks]
\hypersetup{
  colorlinks   = true,    % Colours links instead of ugly boxes
  urlcolor     = blue,    % Colour for external hyperlinks
  linkcolor    = blue,    % Colour of internal links
  citecolor    = ForestGreen      % Colour of citations
}
\usepackage{microtype}
\usepackage[justification=centering]{caption} % [justification=centering]

\usepackage[labelformat=simple]{subcaption}

\DeclareCaptionLabelFormat{parenthesis}{(#2)}
\makeatletter
\renewcommand\p@subfigure{\arabic{figure}.}
\makeatother

\DeclareCaptionLabelFormat{parenthesis}{(#2)}
\makeatletter
\renewcommand\p@subtable{\arabic{table}.}
\makeatother

\usepackage{enumitem}

\setlist[itemize]{leftmargin=2.5\parindent}
\setlist[enumerate]{leftmargin=2.5\parindent}

\def\addlegendimage{\csname pgfplots@addlegendimage\endcsname}

\theoremstyle{plain}
\newtheorem{lemma}{Lemma}[section]

\theoremstyle{definition}
\newtheorem{axiom}{Axiom}%[section]

\newtheorem{definition}{Definition}[section]

\theoremstyle{remark}

\makeatletter
\let\@fnsymbol\@alph
\makeatother

\def\keywords{\vspace{.5em} % Add keywords
{\noindent \textit{Keywords}: }}

\def\AMS{\vspace{.5em} % Add keywords
{\noindent \textbf{\emph{MSC} class}: }}

\def\JEL{\vspace{.5em} % Add keywords
{\noindent \textbf{\emph{JEL} classification number}: }}

\author{\href{https://sites.google.com/view/doragretapetroczy}{D\'ora Gr\'eta Petr\'oczy}\thanks{~E-mail: \emph{doragreta.petroczy@uni-corvinus.hu} \newline Corvinus University of Budapest (BCE), Department of Finance, Budapest, Hungary} $\qquad \qquad$ \href{https://sites.google.com/view/laszlocsato}{L\'aszl\'o Csat\'o}\thanks{~Corresponding author. E-mail: \emph{laszlo.csato@sztaki.hu} \newline Institute for Computer Science and Control (SZTAKI), Laboratory on Engineering and Management Intelligence, Research Group of Operations Research and Decision Systems, Budapest, Hungary \newline % Hungary, 1111 Budapest, Kende street 13-17.
Corvinus University of Budapest (BCE), Department of Operations Research and Actuarial Sciences, Budapest, Hungary}}

\title{Revenue allocation in Formula One: \\ A pairwise comparison approach}
\date{\today}

\def\Dedication{
\begin{small}
{\noindent
\emph{``The analysis of revenue sharing has paid little attention to the different ways that revenues for sharing can be collected or the basis of their allocation.'' \citep[p.~1182]{Szymanski2003}}
}
\end{small}
}

\begin{document}

\maketitle

\Dedication

\begin{abstract}
\noindent
A model is proposed to allocate Formula One World Championship prize money among the constructors. The methodology is based on pairwise comparison matrices, allows for the use of any weighting method, and makes possible to tune the level of inequality.
We introduce an axiom called scale invariance, which requires the ranking of the teams to be independent of the parameter controlling inequality. The eigenvector method is revealed to violate this condition in our dataset, while the row geometric mean method always satisfies it. The revenue allocation is not influenced by the arbitrary valuation given to the race prizes in the official points scoring system of Formula One and takes the intensity of pairwise preferences into account, contrary to the standard Condorcet method. Our approach can be used to share revenues among groups when group members are ranked several times.


\AMS{62F07, 90B50, 91B08}
% Ranking and selection
% Management decision making, including multiple objectives
% Individual preferences

\JEL{C44, D71}
% Operations Research, Statistical Decision Theory
% Social Choice, Clubs, Committees, Associations
\end{abstract}

\clearpage

\section{Introduction} \label{Sec1}

Professional sports leagues and championships generate billions of euros in common revenue. However, its allocation among the participants is often burdened with serious legal disputes centred around unequal shares and the possible violation of competition laws.
Consequently, constructing allocation rules which depend only on a few arbitrary variables, and are relatively simple, robust, and understandable for all participants, poses an important topic of academic research.

Formula One (Formula 1, or simply F1) is the highest class of single-seater car racing. A Formula One season consists of several races taking place around the world. The drivers and constructors accumulate points on each race to obtain the annual World Championships, one for drivers, and one for constructors. Since running a team costs at least 100 million US dollars per year,\footnote{~The study ``Revealed: The \$2.6 billion budget that fuels F1's 10 teams'' by \emph{Christian Sylt} is available at \url{https://www.forbes.com/sites/csylt/2018/04/08/revealed-the-2-6-billion-budget-that-fuels-f1s-ten-teams}.} the distribution of Formula One prize money (1,004 million US dollars in 2019\footnote{~See the summary titled ``Formula 1 teams' prize money payments for 2019 revealed'' by \emph{Dieter Rencken} and \emph{Keith Collantine}, which can be accessed at \url{https://www.racefans.net/2019/03/03/formula-1-teams-prize-money-payments-for-2019-revealed/}.}) can substantially affect competitive balance and the uncertainty around the expected outcome of races.

\begin{table}[ht!]
%\small
\centering
  \rowcolors{1}{}{gray!20}
  \caption{Formula One prize money allocation, 2019 (in million US dollars)}
  \label{Table1}
\begin{threeparttable}
\begin{tabularx}{1\textwidth}{lccCCc >{\bfseries}c}
 \toprule \hiderowcolors
 Team & Column 1 & Column 2 & Long-standing team & Championship bonus & Other & Sum \\ \bottomrule \showrowcolors
    Ferrari & 35    & 56    & 73    & 41    &   ---    & 205 \\
    Mercedes & 35    & 66    &    ---   & 41    & 35    & 177 \\
    Red Bull & 35    & 46    &   ---    & 36    & 35    & 152 \\
    McLaren & 35    & 32    &    ---   & 33    &    ---   & 100 \\
    Renault & 35    & 38    &   ---    &    ---   &   ---    & 73 \\
    Haas  & 35    & 35    &    ---   &   ---    &    ---   & 70 \\
    Williams & 35    & 15    &    ---   &   ---    & 10    & 60 \\
    Racing Point & 35    & 24    &   ---    &  ---    &   ---    & 59 \\
    Sauber & 35    & 21    &    ---   &    ---   &   ---    & 56 \\
    Toro Rosso & 35    & 17    &  ---     &   ---    &    ---   & 52 \\ \bottomrule
    Sum & 350 &350 &73 &151 &80 & 1004 \\ \toprule
\end{tabularx}
\begin{tablenotes} \footnotesize 
\item Source: \url{https://www.racefans.net/2019/03/03/formula-1-teams-prize-money-payments-for-2019-revealed/}
% Rencken and Collantine: \emph{Formula 1 teams' prize money payments for 2019 revealed}, available at
\end{tablenotes}
\end{threeparttable}
\end{table}

The current prize money allocation of Formula One is presented in Table~\ref{Table1}. 
Column 1 corresponds to the revenue distributed equally among the teams which have finished in the top ten in at least two of the past three seasons. Column 2 corresponds to the performance-based payment, determined by the team's finishing position in the previous season. Ferrari has a special Long-standing Team payment as being the only team that competes since the beginning of the championship. Column 4 is paid to the previous champions, and three other teams receive bonus payments.

It is worth noting that only the third of the money pot is allocated strictly on the performance in the previous season and bottom teams are underrepresented.
Therefore, the new owner of the company controlling Formula One since January 2017 (Liberty Media), plans to reform the revenue allocation of the championship, mainly to increase the competitiveness of smaller teams.\footnote{~Consider the article ``F1 2021: Liberty's masterplan for Formula One's future uncovered'' by \emph{Dieter Rencken} and \emph{Keith Collantine}, available at \url{https://www.racefans.net/2018/04/11/f1-2021-libertys-masterplan-for-formula-ones-future-uncovered/}. Another column titled ``Revealed: The winners and losers under Liberty's 2021 F1 prize money plan'' by \emph{Dieter Rencken} and \emph{Keith Collantine} discusses the potential impacts of this plan and can be accessed at \url{https://www.racefans.net/2018/04/11/revealed-the-winners-and-losers-under-libertys-2021-f1-prize-money-plan/}.}

\begin{table}[t]
\centering
\captionsetup{justification=centering}
\caption{Alternative results of the 2013 Formula One World Championship}
\label{Table2}

\begin{subtable}{\linewidth}
\centering
\caption{Points scoring systems}
\label{Table2a}
\rowcolors{1}{gray!20}{}
    \begin{tabularx}{0.8\textwidth}{l CCCCC CCCCC} \toprule \hiderowcolors
    Year  & \multicolumn{10}{c}{Positions} \\
          & 1     & 2     & 3     & 4     & 5     & 6     & 7     & 8     & 9     & 10 \\ \bottomrule \showrowcolors
    1961--1990 & 9     & 6     & 4     & 3     & 2     & 1     & 0     & 0     & 0     & 0 \\
    1991--2002 & 10    & 6     & 4     & 3     & 2     & 1     & 0     & 0     & 0     & 0 \\
    2003--2009 & 10    & 8     & 6     & 5     & 4     & 3     & 2     & 1     & 0     & 0 \\
    2010-- & 25    & 18    & 15    & 12    & 10    & 8     & 6     & 4     & 2     & 1 \\
\bottomrule
    \end{tabularx}
\end{subtable}

\vspace{0.5cm}
\begin{subtable}{\linewidth}
\caption{World Constructors' Championship standings}
\label{Table2b}
\rowcolors{1}{}{gray!20}
\centerline{
    \begin{tabularx}{1.05\textwidth}{l CCCC CC>{\bfseries}C>{\bfseries}C} \toprule \hiderowcolors
    Team  & \multicolumn{8}{c}{Points scoring systems} \\
          & \multicolumn{2}{c}{1961--1990} & \multicolumn{2}{c}{1991--2002} & \multicolumn{2}{c}{2003--2009} & \multicolumn{2}{c}{\textbf{2010--}} \\
          & Points & Rank  & Points & Rank  & Points & Rank  & Points & Rank \\ \bottomrule \showrowcolors
    Red Bull & 190   & 1     & 203   & 1     & 243   & 1     & 596   & 1 \\
    Mercedes & 88    & 2     & 91    & 2     & 140   & 3     & 360   & 2 \\
    Ferrari & 87    & 3     & 89    & 3     & 141   & 2     & 354   & 3 \\
    Lotus & 84    & 4     & 85    & 4     & 127   & 4     & 315   & 4 \\
    McLaren & 12    & 5     & 12    & 5     & 38    & 5     & 122   & 5 \\
    Force India & 6     & 7     & 6     & 7     & 23    & 6     & 77    & 6 \\
    Sauber & 7     & 6     & 7     & 6     & 19    & 7     & 57    & 7 \\
    Toro Rosso & 1     & 8     & 1     & 8     & 9     & 8     & 33    & 8 \\
    Williams & 0     & 9     & 0     & 9     & 1     & 9     & 5     & 9 \\
    Marussia & 0     & 10    & 0     & 10    & 0     & 10    & 0     & 10 \\
    Caterham & 0     & 11    & 0     & 11    & 0     & 11    & 0     & 11 \\ \toprule
    \end{tabularx}
}
\end{subtable}
\end{table}

Nonetheless, since the money allocation is based on the ranking of the constructors, it is important to apply a robust and reliable ranking procedure---but the current points scoring system fails to satisfy this criterion. An illustration is provided by Table~\ref{Table2}: in the 2013 season, Ferrari would have obtained the second position ahead of Mercedes if the previous points scoring system (effective from 2003 to 2009) would have been used, and Sauber would have been the sixth instead of Force India according to the 1961--1990 and 1991--2002 rules.

The long list of Formula One World Championship points scoring systems, available at \url{https://en.wikipedia.org/wiki/List_of_Formula_One_World_Championship_points_scoring_systems}, highlights the arbitrariness of the points scoring systems, and suggests that the relative importance of the different positions in a race remains unclear. The criticism in \citet[Section~2]{Haigh2009} is also worth studying.
A website (\url{https://www.formula1points.com/simulator}) dedicated to trialling arbitrary points system definitions shows the popularity of such what-if scenarios among the fans.

% In a nutshell, we provide a novel methodology based on pairwise comparisons of the constructors in a given season to allocate the revenue generated by the whole championship in the subsequent year.

Inspired by similar discrepancies, the current paper aims to outline a formal model that can be used to share Formula One prize money among the teams in a meaningful way. Our proposal is based on pairwise comparisons and has strong links to the Analytic Hierarchy Process (AHP), a famous decision-making framework. In particular, we construct a multiplicative pairwise comparison matrix from the race results. Contrary to the Condorcet-like methods \citep{SoaresdeMelloetal2015}, it is not only said that a team is preferred to another if it has better results in the majority of the races, but the intensity of these pairwise preferences are also taken into account. Two popular weighting methods, the eigenvector method, and the row geometric mean, are considered to compute the revenue share of each team.

The values of the pairwise comparison matrix depend on a parameter, which is shown to control the inequality of the allocation. Since the ranking of the teams by the eigenvector method is not independent of this parameter, the row geometric mean method has more favourable axiomatic properties in our setting.
The application presented in the paper gives some insight into how the parameter influences the expected inequality of the revenue allocation, thus the decision-maker can fix the rules of the distribution before the competition is conducted.
%developed by Thomas L. Saaty, an influential researcher cited more than 110,000 times according to the web search engine \href{https://scholar.google.com/citations?user=mCAPouEAAAAJ&hl=hu&oi=sra}{GoogleScholar}.

The main contribution of our study is providing an alternative solution in the Formula One industry, which means innovation in reforming revenue allocation and at the same time applicable in other settings to replicate.
The pairwise comparison approach has the following, mainly advantageous, features:
\begin{itemize}
\item
The derivation of the pairwise comparison matrix from the race results depends on a single variable, which regulates the inequality of the distribution. The user can choose this by taking into account preferences on how much inequality is desirable.
\item
It allows for the use of any weight deriving methods used in the AHP literature;
\item
Except for its sole parameter, the methodology is not influenced by any ad hoc decision such as the scores used in the official points system of Formula One.
\item
It supports the reliable performance of the bottom teams, which seldom score points, therefore the current system awards if they achieve unexpected results, mainly due to extreme events in some races.
\end{itemize}
Besides that, a reasonable axiom in our setting is introduced for weighting methods, and its violation by the eigenvector method is presented in real data.

The paper proceeds as follows.
Section~\ref{Sec2} gives an overview of related articles. The methodology is presented in Section~\ref{Sec3} and is applied in Section~\ref{Sec4} to the Formula One World Constructors' Championship between 2014 and 2018. Section~\ref{Sec5} concludes.

\section{Literature review} \label{Sec2}

Our work is connected to several research fields.
First, revenue sharing and its impact on competitive balance is a prominent topic of sports economics. \citet{AtkinsonStanleyTschirhart1988} examine revenue allocation as an incentive mechanism encouraging the desired behaviour of the teams in a league.
\citet{Kesenne2000} analyses revenue sharing under the profit- and utility-maximising hypothesis, and finds that equality promotes competitive balance.
However, revenue sharing can lead to a more uneven contest under reasonable assumptions \citep{Szymanski2003, SzymanskiKesenne2004}.
\citet{BergantinosMoreno-Ternero2020a} give direct, axiomatic, and game-theoretical foundations for two simple rules used to share the money coming from broadcasting sports leagues among participating teams.
\citet{BergantinosMoreno-Ternero2020b} introduce axioms formalising alternative ways of allocating the extra revenue obtained from additional viewers.

Second, the proposed method is based on multiplicative pairwise comparison matrices \citep{Saaty1977, Saaty1980}, thus it continues the applications of the AHP in resource allocation problems \citep{RamanathanGanesh1995, SaatyPeniwatiShang2007}. For example, \citet{Ossadnik1996} extensively uses pairwise comparison matrices to allocate the expected synergies in a merger to the partners.
Furthermore, the current paper offers theoretical contributions to this area. In particular, we consider two popular weighting methods and introduce a reasonable axiom called \emph{scale invariance}, that is, the ranking should be independent of the variable governing the inequality of the allocation. Since it will be investigated whether this requirement is satisfied by the two methods in the case of Formula One results, our paper can be regarded as a companion to \citet{DulebaMoslem2019}, which provides the first examination of another property called \emph{Pareto efficiency} \citep{BlanqueroCarrizosaConde2006, Bozoki2014, BozokiFulop2018} on real data.

Third, there are some direct applications of pairwise comparison matrices in sports ranking. \citet{Csato2013a} and \cite{Csato2017c} recommend this approach to obtain a fair ranking in Swiss-system chess team tournaments. \citet{BozokiCsatoTemesi2016} address the issue of ranking top tennis players of all time, while \citet{ChaoKouLiPeng2018} evaluate historical Go players.

Fourth, many articles have studied economic problems emerging in Formula One.
Rule changes seem to reduce the teams' performances but to improve the competitive balance, and the revenue gain from the latter turns out to be bigger than revenue loss from the former \citep{MastromarcoRunkel2009}.
\citet{JuddeBoothBrooks2013} undertake an econometric analysis of competitive balance in this sport.
\citet{ZaksaiteRadusevicius2017} examine the legal aspects of team orders and other tactical decisions in Formula One.
According to \citet{BudzinskiMuller-Kock2018}, the revenue allocation scheme of this sport is consistent with an anticompetitive interpretation and may be subject to an antitrust investigation.
\citet{GutierrezLozano2018} propose a framework for the design efficiency assessment of some racing circuits that hosted Formula One.
\citet{HendersonKirrane2018} offer a Bayesian approach to forecast race results.
\citet{SchreyerTorgler2018} investigate whether race outcome uncertainty affects the TV demand for Formula One in Germany and conclude that a balanced competition increases the number of viewers.

Fifth, our procedure leads to an alternative ranking of Formula One constructors, which has its antecedents, too.
\citet{Kladroba2000} introduces well-known methods of aggregation to determine the World Championship in 1998.
Some of the ranking problems that occurred in the history of Formula One are found to result from defects of the Borda method \citep{SoaresdeMelloetal2005}.
\citet{Haigh2009} illustrates the instability of the scoring system and argues that any system should be robust to plausible changes, which is not satisfied by the Formula One scoring rules.
According to \citet{Anderson2014}, subjective point-based rankings may fail to provide an accurate ranking of competitors based on ability.
\citet{SoaresdeMelloetal2015} present a variant of the Condorcet method with weakly rational decision-makers to compare the teams which competed in the 2013 season.
\citet{Sitarz2013} presents the incenter of a convex cone to obtain a new ranking of Formula One drivers.
\citet{Phillips2014} measures driver performances by adjusting for team and competition effects.
\citet{BellSmithSabelJones2016} aim to identify the best Formula One drivers of all time.
\citet{Corvalan2018} addresses the problem of how the election of the world champion depends on the particular valuation given to the race prizes.
\citet{KondratevIanovskiNesterov2019} characterise the family of geometric scoring rules by the axioms of consistency for adding/removing a unanimous winner and a unanimous loser, respectively.
\citet{DietzenbacherKondratev2020} is probably the first paper that develops and motivates prize allocation rules for competitions directly from axioms.
%\citet{Csato2020n} investigates a tradeoff of points scoring systems between two risks (the threat of early clinch and the danger of winning without finishing first) on the basis of Formula One data.

\section{Theoretical background} \label{Sec3}

In this section, the main components of the model will be presented: the multiplicative pairwise comparison matrix, its derivation from the race results, a straightforward axiom in our setting, and a basic measure of inequality.

\subsection{Multiplicative pairwise comparison matrices} \label{Sec31}

Consider a set of alternatives $N = \{ 1,2, \dots ,n \}$ such that their pairwise comparisons are known: $a_{ij}$ shows how many times alternative $i$ is better than alternative $j$.

The sets of positive (with all elements greater than zero) vectors of size $n$ and matrices of size $n \times n$ are denoted by the symbols $\mathbb{R}^{n}_+$ and $\mathbb{R}^{n \times n}_+$, respectively.

The pairwise comparisons are collected into a matrix satisfying the reciprocity condition, hence any entry below the diagonal equals the reciprocal of the corresponding entry above the diagonal.

\begin{definition} \label{Def31}
\emph{Multiplicative pairwise comparison matrix}:
Matrix $\mathbf{A} = \left[ a_{ij} \right] \in \mathbb{R}^{n \times n}_+$ is a \emph{multiplicative pairwise comparison matrix} if $a_{ji} = 1/a_{ij}$ holds for all $1 \leq i,j \leq n$.
\end{definition}

We will sometimes omit the word ``multiplicative'' for the sake of simplicity.

Let $\mathcal{A}^{n \times n}$ be the set of pairwise comparison matrices with $n$ alternatives.

Pairwise comparisons are usually used to obtain an approximation of the relative priorities of the alternatives.

\begin{definition} \label{Def32}
\emph{Weight vector}:
Vector $\mathbf{w}  = \left[ w_{i} \right] \in \mathbb{R}^n_+$ is a \emph{weight vector} if $\sum_{i=1}^n w_{i} = 1$.
\end{definition}

Let $\mathcal{R}^{n}$ be the set of weight vectors of size $n$.

\begin{definition} \label{Def33}
\emph{Weighting method}:
Mapping $f: \mathcal{A}^{n \times n} \to \mathcal{R}^{n}$ is a \emph{weighting method}.
\end{definition}

The weight of alternative $i$ in the pairwise comparison matrix $\mathbf{A} \in \mathcal{A}^{n \times n}$ according to the weighting method $f$ is denoted by $f_i(\mathbf{A})$.

There exist several methods to derive a weight vector, see, for example, \citet{ChooWedley2004} for a thorough overview.
The most popular procedures are the \emph{eigenvector method} \citep{Saaty1977, Saaty1980}, and the \emph{row geometric mean (logarithmic least squares) method} \citep{WilliamsCrawford1980, CrawfordWilliams1985, DeGraan1980, deJong1984, Rabinowitz1976}.

\begin{definition} \label{Def34}
\emph{Eigenvector method} ($EM$):
The weight vector $\mathbf{w}^{EM} (\mathbf{A}) \in \mathcal{R}^n$ provided by the \emph{eigenvector method} is the solution of the following system of linear equations for any pairwise comparison matrix $\mathbf{A} \in \mathcal{A}^{n \times n}$:
\[
\mathbf{A} \mathbf{w}^{EM}(\mathbf{A}) = \lambda_{\max} \mathbf{w}^{EM}(\mathbf{A}),
\]
where $\lambda_{\max}$ denotes the maximal eigenvalue, also known as the principal or Perron eigenvalue, of the (positive) matrix $\mathbf{A}$.
\end{definition}

\begin{definition} \label{Def35}
\emph{Row geometric mean method} ($RGM$):
The \emph{row geometric mean method} is the function $\mathbf{A} \to \mathbf{w}^{RGM} (\mathbf{A})$ such that the weight vector $\mathbf{w}^{RGM} (\mathbf{A})$ is given by
\[
w_i^{RGM}(\mathbf{A}) = \frac{\prod_{j=1}^n a_{ij}^{1/n}}{\sum_{k=1}^n \prod_{j=1}^n a_{kj}^{1/n}}.
\]
\end{definition}

The row geometric mean method is sometimes called the \emph{Logarithmic Least Squares Method} ($LLSM$) because it is the solution to the following optimisation problem:
\[
\min_{\mathbf{w} \in \mathcal{R}^n} \sum_{i=1}^n \sum_{j=1}^n \left[ \log a_{ij} - \log \left( \frac{w_i}{w_j} \right) \right]^2.
\]

Although the application of the row geometric mean is axiomatically well-founded \citep{Fichtner1986, Barzilai1997, LundySirajGreco2017, BozokiTsyganok2019}, and the eigenvector method has some serious theoretical shortcomings \citep{JohnsonBeineWang1979, BlanqueroCarrizosaConde2006, BanaeCostaVansnick2008}, Saaty's proposal remains the default choice of most practitioners.
Therefore, both procedures will be considered.

\subsection{From race results to a pairwise comparison matrix} \label{Sec32}

A Formula One season consists of a series of races, contested by two cars/drivers of each constructor/team. We say that team $i$ has scored one goal against team $j$ if a given car of team $i$ is ahead of a given car of team $j$ in a race. Thus, if there are no incomparable cars, then: 
\begin{itemize}
\item
Team $i$ ($j$) has scored four (zero) goals against team $j$ ($i$) if both cars of team $i$ have finished above both cars of team $j$;
\item
Team $i$ ($j$) has scored three (one) goal(s) against team $j$ ($i$) if one car of team $i$ has finished above both cars of team $j$, and the other car of team $i$ has finished above one car of team $j$;
\item
Team $i$ ($j$) has scored two (two) goals against team $j$ ($i$) if one car of team $i$ has finished above both cars of team $j$ but both cars of team $j$ have finished above the other car of team $i$.\footnote{~We follow the official definition in classifying a driver as finished if he completed over 90\% of the race distance.}
\end{itemize}
The goals scored by the constructors in a race are aggregated over the whole season without weighting, similarly to the official points scoring system.\footnote{~Except for the \href{https://en.wikipedia.org/wiki/2014_Abu_Dhabi_Grand_Prix}{2014 Abu Dhabi Grand Prix}, which awarded double points as the last race of the season.}

Consequently, the maximum number of goals that a team can score against another is four times the number of races. Since a car might not finish a race, it is assumed that each finishing car is better than another, which fails to finish. Nonetheless, two cars may be incomparable if both of them failed to finish the race. In this case, no goal is scored.
The goals of the constructors are collected into the $n \times n$ \emph{goals matrix}.

The pairwise comparison matrix $\mathbf{A} = \left[ a_{ij} \right]$ is obtained from the goals matrix: if constructor $i$ has scored $g_{ij}$ goals against constructor $j$, while constructor $j$ has scored $g_{ji}$ goals against constructor $i$, then $a_{ij} = g_{ij} / g_{ji}$ and $a_{ji} = g_{ji} / g_{ij}$ to guarantee the reciprocity condition.

Theoretically, this procedure is ill-defined because the problem of division by zero is not addressed. However, in our dataset $g_{ij}$ was always positive, in other words, at least one car of every team was ahead of one car of any other team at least in one race during the whole season. Thus the somewhat arbitrary adjustment of zeros in \citet{BozokiCsatoTemesi2016} is avoided here. Nonetheless, the potential problem cannot be neglected because it did not appear in our dataset. For us, adding a constant $\varepsilon$ to each element of the goals matrix is an attractive solution as the definition $a_{ij} = (g_{ij} + \varepsilon) / (g_{ji} + \varepsilon)$ also mitigates the sharp differences between the top and bottom teams, which greatly depend on whether the drivers of the weak team were ``lucky'' enough to finish certain races. However, the thorough analysis of this suggestion is left to future research.

As we have mentioned, two weighting methods, the $EM$ and the $RGM$ will be used to derive a weight vector, which directly provides an allocation of the available amount.

The presented procedure does not contain any variable, thus it might lead to an allocation that cannot be tolerated by the decision-maker because of its (in)equality.
Hence the definition of the pairwise comparison matrix is modified such that:
\[
a_{ij} = \left( \frac{g_{ij}}{g_{ji}} \right)^\alpha \qquad \text{and} \qquad a_{ji} = \left( \frac{g_{ji}}{g_{ij}} \right)^\alpha \qquad\text{for all $i, j \in N$},
\]
where $\alpha \geq 0$ is a parameter. If $\alpha$ is small, then $\mathbf{A}$ is close to the unit matrix, the weights are almost the same, and the shares remain roughly equal. The effect of $\alpha$ will be further investigated in Section~\ref{Sec4}.

Note that the proposed model is similar to the Condorcet-like methods, widely used in social choice theory, only at first sight. For example, \citet{SoaresdeMelloetal2015} introduces a so-called Condorcet graph, where the nodes represent the Formula One teams, and there is a directed edge from node $i$ to node $j$ if team $i$ is preferred to team $j$, that is, if team $i$ is preferred to team $j$ in the majority of races. Therefore, the Condorcet variants take only the pairwise preferences into account but not their intensities. Unsurprisingly, this loss of information often means that the Condorcet method does not provide a strict ranking, while ties between the revenue share of two teams in our model are less frequent (and does not mean any problem). 

\subsection{A natural axiom for weighting methods} \label{Sec33}

%In the following, some theoretical results concerning the parameter $\alpha$ are presented.
Some papers \citep{GenestLapointeDrury1993, CsatoRonyai2016} examine ordinal pairwise preferences, that is, pairwise comparison matrices with entries of $a$ or $1/a$.
This idea has inspired the following requirement, which is also an adaptation of the property called \emph{power invariance} \citep{Fichtner1984} for the ranking of the alternatives.

\begin{axiom}
\emph{Scale invariance}:
Let $\mathbf{A} = \left[ a_{ij} \right] \in \mathcal{A}^{n \times n}$ be any pairwise comparison matrix and $\alpha > 0$ be a (positive) parameter. Let $\mathbf{A}^{(\alpha)} = \left[ a_{ij}^{(\alpha)} \right] \in \mathcal{A}^{n \times n}$ be the pairwise comparison matrix defined by $a_{ij}^{(\alpha)} = a_{ij}^{\alpha}$. 
Weighting method $f: \mathcal{A}^{n \times n} \to \mathcal{R}^n$ is called \emph{scale invariant} if for all $1 \leq i,j \leq n$:
\[
f_i \left( \mathbf{A} \right) \geq f_j \left( \mathbf{A} \right) \iff f_i \left( \mathbf{A}^{(\alpha)} \right) \geq f_j \left( \mathbf{A}^{(\alpha)} \right).
\]
\end{axiom}

Scale invariance implies that the ranking of the alternatives does not change if a different scale is used for pairwise comparisons. For example, when only two verbal expressions, ``weakly preferred'' and ``strongly preferred'' are allowed, the ranking should remain the same if these preferences are represented by the values $2$ and $3$, or $4$ and $9$, respectively.

%\citet{Csato2017c} and \citet{Csato2018c} introduce analogous conditions in different settings.

In the setting of Section~\ref{Sec32}, scale invariance does not allow the ranking of the teams to depend on the parameter $\alpha$, which seems to be reasonable because the underlying data (the goals matrix) are fixed. In other words, if constructor $i$ receives more money than constructor $j$ under any value of $\alpha$, then it should receive more money for all $\alpha > 0$.

\begin{lemma} \label{Lemma31}
The eigenvector method does not satisfy scale invariance.
\end{lemma}

\begin{proof}
It is sufficient to provide a counterexample.
See \citet[Example~2.1]{GenestLapointeDrury1993} and \citet[Figure~1]{GenestLapointeDrury1993}, which are based on an example adapted from \citet{Kendall1955}.
\end{proof}

The following result has already been mentioned in \citet[p.~581-582]{GenestLapointeDrury1993}.

\begin{lemma} \label{Lemma32}
The row geometric mean method satisfies scale invariance.
\end{lemma}

\begin{proof}
Note that $w_i^{RGM} \geq w_j^{RGM} \iff \prod_{k=1}^n a_{ik} \geq \prod_{k=1}^n a_{jk} \iff \prod_{k=1}^n a_{ik}^{(\alpha)} \geq \prod_{k=1}^n a_{jk}^{(\alpha)}$, which immediately verifies the statement. 
\end{proof}

\subsection{Measuring inequality} \label{Sec34}

A basic indicator of competition among firms is the Herfindahl--Hirschman index \citep{HannahKay1977, Rhoades1993, Laine1995, OwenRyanWeatherston2007}.
%This will be adopted to measure the (in)equality of the suggested revenue allocation.

\begin{definition} \label{Def36}
\emph{Herfindahl--Hirschman index} ($HHI$):
Let $\mathbf{w} \in \mathcal{R}^n$ be a weight vector. Its \emph{Herfindahl--Hirschman index} is:
\[
HHI(\mathbf{w}) = \sum_{i=1}^n w_i^2.
\]
\end{definition}
The maximum of $HHI$ is one when one constructor receives the whole amount. However, its minimum (reached when all constructors receive the same amount) is influenced by the number of constructors, therefore it is worth considering a normalised version of the $HHI$.

\begin{definition} \label{Def37}
\emph{Normalised Herfindahl--Hirschman index} ($HHI^*$):
Let $\mathbf{w} \in \mathcal{R}^n$ be a weight vector. Its \emph{normalised Herfindahl--Hirschman index} is:
\[
HHI^*(\mathbf{w}) = \frac{HHI(\mathbf{w}) - 1/n}{1 - 1/n} = \frac{\sum_{i=1}^n w_i^2 - 1/n}{1 - 1/n}.
\]
\end{definition}
The value of $HHI^*$ is always between $0$ (equal shares) and $1$ (maximal inequality).

Since the $HHI$ better reflects market concentration, while the normalised Herfindahl--Hirschman index quantifies the equality of distributions, we will use the latter.

\section{Results} \label{Sec4}

\begin{table}[t]
  \centering
  \caption{Goals matrix, 2014}
  \label{Table3}
    \rowcolors{3}{}{gray!20}
    \begin{tabularx}{\textwidth}{l CCCC CCCC CCC} \toprule \hiderowcolors
          &  \rotatebox{90}{Mercedes} &  \rotatebox{90}{Red Bull} &  \rotatebox{90}{Williams} &  \rotatebox{90}{Ferrari} &  \rotatebox{90}{McLaren} &  \rotatebox{90}{Force India} &  \rotatebox{90}{Toro Rosso} &  \rotatebox{90}{Lotus} &  \rotatebox{90}{Marussia} &  \rotatebox{90}{Sauber} &  \rotatebox{90}{Caterham} \\ \bottomrule \showrowcolors
     Mercedes & ---   & 61    & 63    & 62    & 64    & 64    & 65    & 65    & 66    & 66    & 66 \\
     Red Bull & 14    & ---   & 43    & 54    & 56    & 58    & 64    & 64    & 64    & 64    & 64 \\
     Williams & 11    & 31    & ---   & 44    & 47    & 47    & 62    & 66    & 68    & 68    & 68 \\
     Ferrari & 13    & 22    & 31    & ---   & 42    & 48    & 64    & 66    & 69    & 70    & 69 \\
     McLaren & 11    & 20    & 29    & 34    & ---   & 41    & 59    & 65    & 69    & 70    & 71 \\
     Force India & 12    & 16    & 28    & 28    & 35    & ---   & 52    & 61    & 64    & 63    & 64 \\
     Toro Rosso & 10    & 7     & 11    & 12    & 16    & 21    & ---   & 47    & 56    & 54    & 56 \\
     Lotus & 5     & 8     & 6     & 10    & 11    & 11    & 22    & ---   & 49    & 40    & 51 \\
     Marussia & 6     & 8     & 7     & 7     & 7     & 7     & 15    & 16    & ---   & 28    & 34 \\
     Sauber & 7     & 8     & 4     & 5     & 2     & 5     & 13    & 29    & 44    & ---   & 44 \\
     Caterham & 3     & 8     & 5     & 6     & 3     & 4     & 15    & 12    & 22    & 23    & --- \\ \toprule
    \end{tabularx}
\end{table}

To illustrate the proposed allocation scheme, the 2014 Formula One season will be investigated in detail.
Table~\ref{Table3} shows the goals matrix where the teams are ranked according to the official championship result (Sauber and Marussia both obtained zero points this year). For example, a car of Mercedes was better than a car of Red Bull on $61$ occasions, while a car of Red Bull was better than a car of Mercedes on $14$ occasions.

\begin{table}[t]
  \centering
  \caption{Pairwise comparison matrix, $\alpha = 1$, 2014}
  \label{Table4}
    \rowcolors{3}{}{gray!20}
\centerline{
    \begin{tabularx}{1.1\textwidth}{l CCCC CCCC CCC} \toprule \hiderowcolors
          &  \rotatebox{90}{Mercedes} &  \rotatebox{90}{Red Bull} &  \rotatebox{90}{Williams} &  \rotatebox{90}{Ferrari} &  \rotatebox{90}{McLaren} &  \rotatebox{90}{Force India} &  \rotatebox{90}{Toro Rosso} &  \rotatebox{90}{Lotus} &  \rotatebox{90}{Marussia} &  \rotatebox{90}{Sauber} &  \rotatebox{90}{Caterham} \\ \bottomrule \showrowcolors
    Mercedes & 1     & 4.36  & 5.73  & 4.77  & 5.82  & 5.33  & 6.5   & 13    & 11    & 9.43  & 22 \\
    Red Bull & 0.23  & 1     & 1.39  & 2.45  & 2.8   & 3.63  & 9.14  & 8     & 8     & 8     & 8 \\
    Williams & 0.17  & 0.72  & 1     & 1.42  & 1.62  & 1.68  & 5.64  & 11    & 9.71  & 17    & 13.6 \\
    Ferrari & 0.21  & 0.41  & 0.70  & 1     & 1.24  & 1.71  & 5.33  & 6.6   & 9.86  & 14    & 11.5 \\
    McLaren & 0.17  & 0.36  & 0.62  & 0.81  & 1     & 1.17  & 3.69  & 5.91  & 9.86  & 35    & 23.67 \\
    Force India & 0.19  & 0.28  & 0.60  & 0.58  & 0.85  & 1     & 2.48  & 5.55  & 9.14  & 12.6  & 16 \\
    Toro Rosso & 0.15  & 0.11  & 0.18  & 0.19  & 0.27  & 0.40  & 1     & 2.14  & 3.73  & 4.15  & 3.73 \\
    Lotus & 0.08  & 0.13  & 0.09  & 0.15  & 0.17  & 0.18  & 0.47  & 1     & 3.06  & 1.38  & 4.25 \\
    Marussia & 0.09  & 0.13  & 0.10  & 0.10  & 0.10  & 0.11  & 0.27  & 0.33  & 1     & 0.64  & 1.55 \\
    Sauber & 0.11  & 0.13  & 0.06  & 0.07  & 0.03  & 0.08  & 0.24  & 0.73  & 1.57  & 1     & 1.91 \\
    Caterham & 0.05  & 0.13  & 0.07  & 0.09  & 0.04  & 0.06  & 0.27  & 0.24  & 0.65  & 0.52  & 1 \\ \toprule
    \end{tabularx}
}
\end{table}

The corresponding pairwise comparison is presented in Table~\ref{Table4} with the choice $\alpha = 1$. It can be seen that all entries above the diagonal are higher than the corresponding element below the diagonal, thus a team, which has scored more points in the official ranking, is almost always preferred to a team with a lower number of points by pairwise comparisons. The only exception is Marussia vs.\ Sauber, where the latter constructor has a robust advantage.

\input{Figure1_money_allocation2014}

Figure~\ref{Fig1} plots the shares of the $11$ competing teams with our methodology. Due to the same underlying pairwise comparison matrix, results given by the eigenvector and row geometric mean methods do not differ substantially. Nonetheless, Mercedes, McLaren, and Sauber are always better off with the eigenvector method, while Red Bull prefers it if $\alpha$ does not exceed $1.25$.

Mercedes is the dominant team according to both methods, with an ever increasing share as a function of parameter $\alpha$ (see Figure~\ref{Fig1a}).
There are five middle teams (Red Bull, Williams, Ferrari, McLaren, Force India), which receive more money than the equal share in the case of small $\alpha$ (see Figure~\ref{Fig1b}). Their ranking coincides with the official one for the row geometric mean, although Ferrari, which has scored $216$ points in the season, has only a negligible advantage over McLaren, which has scored $181$ points.
On the other hand, the five weakest teams (Toro Rosso, Lotus, Marussia, Sauber, Caterham) receive less if parameter $\alpha$ increases (see Figure~\ref{Fig1c}). The row geometric mean ranking reflects the official ranking again, although the difference between the shares of Marussia and Sauber is barely visible.

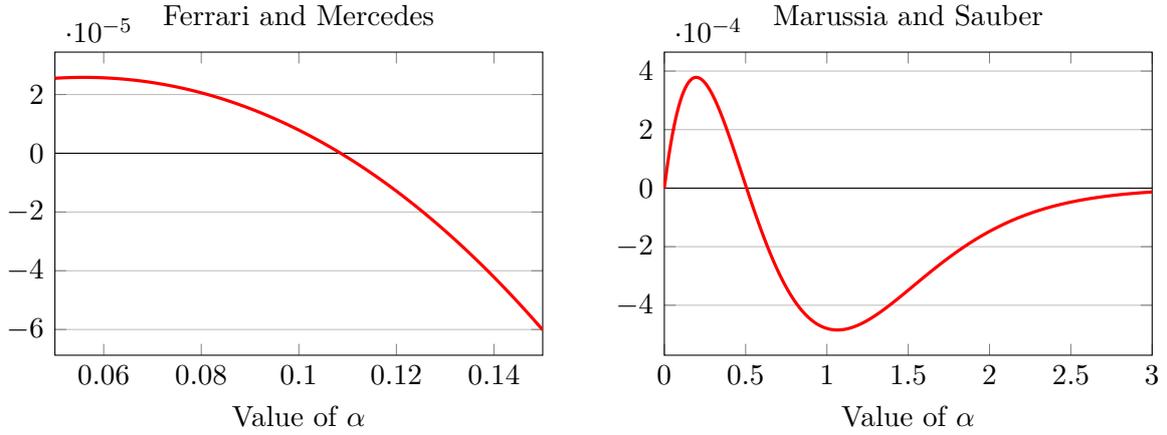
\begin{figure}[ht!]
\centering
\caption{The difference between the revenue shares of some constructors \newline according to the eigenvector method ($EM$), 2014}
\label{Fig2}

\begin{tikzpicture}
%\selectcolormodel{gray}
\begin{axis}[
name = axis1,
width = 0.5\textwidth, 
height = 0.35\textwidth,
title = {Ferrari and Mercedes},
title style = {align=center, font=\small},
xmin = 0.05,
xmax = 0.15,
ymajorgrids,
xlabel = Value of $\alpha$,
xlabel style = {font=\small},
%ylabel = Probability decrease compared to $RR$,
tick label style = {/pgf/number format/fixed},
]
\draw (axis cs:\pgfkeysvalueof{/pgfplots/xmin},0)  -- (axis cs:\pgfkeysvalueof{/pgfplots/xmax},0);
\addplot[red,smooth,very thick] coordinates {
(0.05,2.55287148678007E-05)
(0.051,2.56233134104111E-05)
(0.052,2.57005843280894E-05)
(0.053,2.57604691454028E-05)
(0.054,2.58029093971046E-05)
(0.055,2.58278466301054E-05)
(0.056,2.58352224031949E-05)
(0.057,2.58249782875003E-05)
(0.058,2.57970558682069E-05)
(0.059,2.57513967427125E-05)
(0.06,2.56879425239998E-05)
(0.061,2.56066348389983E-05)
(0.062,2.55074153297091E-05)
(0.063,2.53902256543981E-05)
(0.064,2.52550074869018E-05)
(0.065,0.000025101702518307)
(0.066,2.49302524567979E-05)
(0.067,2.47405990282945E-05)
(0.068,2.45326839765914E-05)
(0.069,0.000024306449064801)
(0.07,2.40618360752981E-05)
(0.071,2.37987868102058E-05)
(0.072,2.35172430914921E-05)
(0.073,0.000023217146762497)
(0.074,2.28984396878906E-05)
(0.075,2.25610637538953E-05)
(0.076,2.22049608692987E-05)
(0.077,2.18300729656895E-05)
(0.078,2.14363419979019E-05)
(0.079,2.10237099453892E-05)
(0.08,2.05921188106001E-05)
(0.081,2.01415106222957E-05)
(0.082,1.96718274338981E-05)
(0.083,1.91830113250024E-05)
(0.084,0.000018675004401697)
(0.085,1.81477487966042E-05)
(0.086,1.76011866703102E-05)
(0.087,1.70352602111012E-05)
(0.088,1.64499116359074E-05)
(0.089,1.58450831903023E-05)
(0.09,1.52207171497942E-05)
(0.091,1.45767558195065E-05)
(0.092,1.39131415350935E-05)
(0.093,1.32298166636013E-05)
(0.094,1.25267236028986E-05)
(0.095,1.18038047833974E-05)
(0.096,1.10610026679003E-05)
(0.097,1.02982597520035E-05)
(0.098,9.51551856519239E-06)
(0.099,8.71272167068948E-06)
(0.1,7.88981166630065E-06)
(0.101,7.04673118480381E-06)
(0.102,6.18342289490648E-06)
(0.103,5.29982950049634E-06)
(0.104,4.39589374280625E-06)
(0.105,3.47155839959534E-06)
(0.106,2.52676628670345E-06)
(0.107,1.56146025719062E-06)
(0.108,5.75583203099583E-07)
(0.109,-4.30921944197249E-07)
(0.11,-1.45811221449743E-06)
(0.111,-2.50604459729742E-06)
(0.112,-3.57477604130685E-06)
(0.113,-4.66436345460119E-06)
(0.114,-5.77486370369196E-06)
(0.115,-6.90633361290216E-06)
(0.116,-8.05882996431084E-06)
(0.117,-9.2324094971008E-06)
(0.118,-1.04271289061986E-05)
(0.119,-1.16430448439953E-05)
(0.12,-1.28802139170991E-05)
(0.121,-1.41386926876957E-05)
(0.122,-1.54185376725902E-05)
(0.123,-1.67198053426937E-05)
(0.124,-1.80425521224964E-05)
(0.125,-1.93868343889986E-05)
(0.126,-2.07527084727932E-05)
(0.127,-2.21402306562896E-05)
(0.128,-2.35494571728112E-05)
(0.129,-2.49804442077894E-05)
(0.13,-2.64332478966955E-05)
(0.131,-2.79079243257069E-05)
(0.132,-2.94045295301942E-05)
(0.133,-3.09231194950133E-05)
(0.134,-3.24637501536024E-05)
(0.135,-3.40264773879967E-05)
(0.136,-3.56113570268018E-05)
(0.137,-3.72184448499957E-05)
(0.138,-3.88477965770073E-05)
(0.139,-4.04994678809972E-05)
(0.14,-4.21735143759927E-05)
(0.141,-4.38699916239932E-05)
(0.142,-4.55889551310013E-05)
(0.143,-4.73304603480079E-05)
(0.144,-4.90945626689943E-05)
(0.145,-5.08813174329992E-05)
(0.146,-5.26907799200116E-05)
(0.147,-5.45230053550089E-05)
(0.148,-5.63780489010041E-05)
(0.149,-5.82559656679971E-05)
(0.15,-6.01568107019973E-05)
};
\end{axis}

\begin{axis}[
at = {(axis1.south east)},
xshift = 0.1\textwidth,
width = 0.5\textwidth, 
height = 0.35\textwidth,
title = {Marussia and Sauber},
title style = {align=center, font=\small},
xmin = 0,
xmax = 3,
ymajorgrids,
xlabel = Value of $\alpha$,
xlabel style = {font=\small},
%ylabel = Probability decrease compared to $RR$,
tick label style = {/pgf/number format/fixed},
]
\draw (axis cs:\pgfkeysvalueof{/pgfplots/xmin},0)  -- (axis cs:\pgfkeysvalueof{/pgfplots/xmax},0);
\addplot[red,smooth,very thick] coordinates {
(0,0)
(0.02,0.000082380946818103)
(0.04,0.000152878465415895)
(0.06,0.000212268813577998)
(0.08,0.000261312370353095)
(0.1,0.000300751864823895)
(0.12,0.000331310727236009)
(0.14,0.000353691572824999)
(0.16,0.000368574826960399)
(0.18,0.000376617498446294)
(0.2,0.000378452106009411)
(0.22,0.000374685761209498)
(0.24,0.000365899409229299)
(0.26,0.000352647227306102)
(0.28,0.000335456178949804)
(0.3,0.000314825720585901)
(0.32,0.000291227655889702)
(0.34,0.000265106131846896)
(0.36,0.000236877769497802)
(0.38,0.000206931921400803)
(0.4,0.000175631047103698)
(0.42,0.000143311197311598)
(0.44,0.000110282597012901)
(0.46,7.68303175385995E-05)
(0.48,0.000043215027398201)
(0.5,9.67381172409937E-06)
(0.52,-2.35789497169958E-05)
(0.54,-5.63506557716004E-05)
(0.56,-8.84695179041023E-05)
(0.58,-0.000119783564059502)
(0.6,-0.0001501596210156)
(0.62,-0.000179482283334899)
(0.64,-0.0002076528764616)
(0.66,-0.000234588420918599)
(0.68,-0.000260220603949101)
(0.7,-0.000284494764331501)
(0.72,-0.000307368895474603)
(0.74,-0.000328812671303198)
(0.76,-0.000348806498841499)
(0.78,-0.000367340600840697)
(0.8,-0.0003844141312522)
(0.82,-0.000400034325830401)
(0.84,-0.000414215689675302)
(0.86,-0.000426979223071399)
(0.88,-0.000438351686574601)
(0.9,-0.000448364905920401)
(0.92,-0.000457055116988399)
(0.94,-0.0004644623507553)
(0.96,-0.000470629857898299)
(0.98,-0.0004756035724728)
(1,-0.000479431613887899)
(1.02,-0.0004821638262201)
(1.04,-0.0004838513537653)
(1.06,-0.00048454625160043)
(1.08,-0.000484301129831551)
(1.1,-0.00048316883012659)
(1.12,-0.0004812021330729)
(1.14,-0.000478453494860739)
(1.16,-0.00047497481176958)
(1.18,-0.00047081721092358)
(1.2,-0.00046603086578514)
(1.22,-0.000460664834868289)
(1.24,-0.00045476692217537)
(1.26,-0.00044838355789082)
(1.28,-0.00044155969790238)
(1.3,-0.00043433874076191)
(1.32,-0.000426762460744559)
(1.34,-0.00041887095571528)
(1.36,-0.0004107026085641)
(1.38,-0.00040229406102647)
(1.4,-0.0003936801987607)
(1.42,-0.00038489414661168)
(1.44,-0.00037596727304652)
(1.46,-0.0003669292028048)
(1.48,-0.00035780783686232)
(1.5,-0.00034862937886238)
(1.52,-0.00033941836722286)
(1.54,-0.00033019771218019)
(1.56,-0.00032098873708256)
(1.58,-0.00031181122329417)
(1.6,-0.00030268345812023)
(1.62,-0.00029362228520802)
(1.64,-0.00028464315692365)
(1.66,-0.00027576018824557)
(1.68,-0.00026698621175624)
(1.7,-0.00025833283335097)
(1.72,-0.00024981048831895)
(1.74,-0.00024142849748527)
(1.76,-0.00023319512313502)
(1.78,-0.00022511762446973)
(1.8,-0.00021720231237584)
(1.82,-0.00020945460330979)
(1.84,-0.0002018790721301)
(1.86,-0.00019447950372856)
(1.88,-0.000187258943334847)
(1.9,-0.000180219745387725)
(1.92,-0.000173363620884339)
(1.94,-0.000166691683135954)
(1.96,-0.000160204491873332)
(1.98,-0.000153902095659352)
(2,-0.00014778407257868)
(2.02,-0.000141849569186232)
(2.04,-0.000136097337706347)
(2.06,-0.000130525771483976)
(2.08,-0.000125132938697599)
(2.1,-0.000119916614350962)
(2.12,-0.000114874310567337)
(2.14,-0.000110003305215748)
(2.16,-0.000105300668903626)
(2.18,-0.000100763290374645)
(2.2,-0.00009638790035419)
(2.22,-0.000092171093887938)
(2.24,-0.00008810935122157)
(2.26,-0.000084199057271654)
(2.28,-0.000080436519739253)
(2.3,-0.000076817985919014)
(2.32,-0.000073339658257143)
(2.34,-0.000069997708712208)
(2.36,-0.000066788291972685)
(2.38,-0.000063707557585079)
(2.4,-0.000060751661045965)
(2.42,-0.000057916773910698)
(2.44,-0.00005519909297066)
(2.46,-0.000052594848549957)
(2.48,-0.000050100311971301)
(2.5,-0.000047711802239569)
(2.52,-0.000045425691990147)
(2.54,-0.000043238412747717)
(2.56,-0.000041146459539632)
(2.58,-0.000039146394906448)
(2.6,-0.00003723485235054)
(2.62,-3.54085392621431E-05)
(2.64,-3.36642393604687E-05)
(2.66,-3.19988146858568E-05)
(2.68,-3.04092071773513E-05)
(2.7,-2.88924398683413E-05)
(2.72,-2.74456177313698E-05)
(2.74,-2.60659282015316E-05)
(2.76,-2.47506414063784E-05)
(2.78,-2.34971101286589E-05)
(2.8,-2.23027695267511E-05)
(2.82,-2.11651366361749E-05)
(2.84,-2.00818096741656E-05)
(2.86,-1.90504671679181E-05)
(2.88,-1.80688669257983E-05)
(2.9,-1.71348448695482E-05)
(2.92,-1.62463137442943E-05)
(2.94,-1.54012617220037E-05)
(2.96,-1.45977509129162E-05)
(2.98,-1.38339157984167E-05)
(3,-1.31079615978035E-05)
};
\end{axis}
\end{tikzpicture}

\end{figure}
%\end{document}

Lemma~\ref{Lemma32} guarantees that the ranking of the constructors according to their share of the revenue remains unchanged as a function of $\alpha$ if the row geometric mean method is applied. On the contrary, the eigenvector method can lead to a rank reversal: Figure~\ref{Fig1b} shows that McLaren will receive a higher share than Williams if $\alpha$ is bigger than $1.5$. There are two other changes in the ranking as uncovered by Figure~\ref{Fig2}: Ferrari receives more money than McLaren, and Marussia receives more money than Sauber for small values of the parameter. Note that row geometric mean always favours the former teams, Ferrari and Marussia.

Thus the violation of scale invariance by the eigenvector method---similarly to the violation of Pareto efficiency \citep{DulebaMoslem2019}---is not only a theoretical curiosity (Lemma~\ref{Lemma31}) as real data, in our case Formula One results from the 2014 season, may result in such an undesired rank reversal.

The theoretical explanation of these results remains challenging. It is probably sufficient to cite \citet[p.~29]{HermanKoczkodaj1996}: ``\emph{It is improbable that an analytical solution can be devised in a situation where the results favour GM over EV for one metric while favouring EV over GM for another metric.}''
At least, according to \citet[Theorem~9]{SaatyVargas1984c}, both weighting methods satisfy \emph{row dominance}, namely, for any pair of alternatives $i,j$, if $a_{ik} \geq a_{jk}$ holds for all $k$, then $w_i \geq w_j$. For example, this means that Force India should receive at least as much money as Toro Rosso (take a look at Table~\ref{Table4}).
Recently, \citet{KulakowskiMazurekStrada2020} have provided lower and upper bounds with the use of the Koczkodaj inconsistency index \citep{Koczkodaj1993, DuszakKoczkodaj1994} for the maximal difference of the weights derived by the eigenvector and the geometric mean methods. In addition, the two priority vectors are proved to be more similar if the pairwise comparison matrix is less inconsistent \citep[Theorem~19]{KulakowskiMazurekStrada2020}, which corresponds to a lower $\alpha$ in our case. A deeper analysis of this issue is the subject of future research.

\begin{table}[t]
\caption{Alternative results of the 2014 Formula One World Championship}
\label{Table5}
\rowcolors{1}{}{gray!20}
\centerline{
    \begin{tabularx}{1.05\textwidth}{l CCCC CC>{\bfseries}C>{\bfseries}C} \toprule \hiderowcolors
    Team  & \multicolumn{8}{c}{Points scoring systems} \\
          & \multicolumn{2}{c}{1961--1990} & \multicolumn{2}{c}{1991--2002} & \multicolumn{2}{c}{2003--2009} & \multicolumn{2}{c}{\textbf{2010--}} \\
          & Points & Rank  & Points & Rank  & Points & Rank  & Points & Rank \\ \bottomrule \showrowcolors
    Mercedes & 233   & 1     & 249   & 1     & 281   & 1     & 676   & 1 \\
    Red Bull & 96    & 2     & 99    & 2     & 154   & 2     & 389   & 2 \\
    Williams & 62    & 3     & 62    & 3     & 113   & 3     & 287   & 3 \\
    Ferrari & 36    & 4     & 36    & 4     & 78    & 4     & 213   & 4 \\
    McLaren & 31    & 5     & 31    & 5     & 62    & 5     & 171   & 5 \\
    Force India & 16    & 6     & 16    & 6     & 46    & 6     & 141   & 6 \\
    Toro Rosso & 1     & 7     & 1     & 7     & 5     & 7     & 30    & 7 \\
    Lotus & 0     & 8     & 0     & 8     & 2     & 8     & 10    & 8 \\
    Marussia & 0     & 9     & 0     & 9     & 0     & 9     & 2     & 9 \\
    Sauber & 0     & 10    & 0     & 10    & 0     & 10    & 0     & 10 \\
    Caterham & 0     & 11    & 0     & 11    & 0     & 11    & 0     & 11 \\ \toprule
    \end{tabularx}
}
\end{table}

Table~\ref{Table5} presents the standing in the 2014 season under the four basic historical points scoring systems that are detailed in Table~\ref{Table2b}. The official 2014 scheme, which awarded double points in the last race, is not included, however, it leads to the same ranking as the other four rules. The ranking is identical to the one determined by the row geometric mean method. This can be another argument against the use of the eigenvector method.

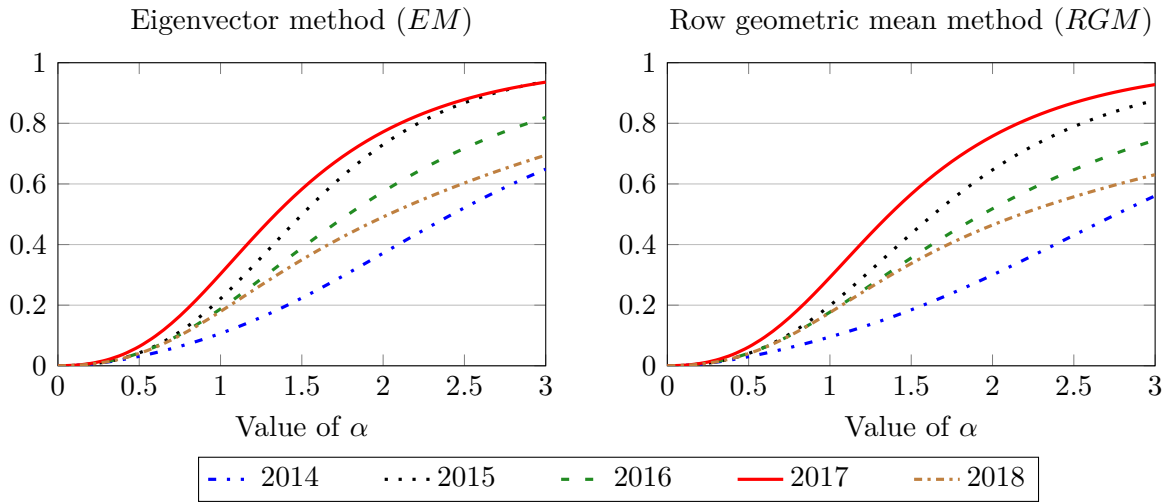
\begin{figure}[ht!]
\centering
\caption{The normalised Herfindahl--Hirschman index ($HHI^*$), 2014--2018}
\label{Fig3}

\begin{tikzpicture}
%\selectcolormodel{gray}
\begin{axis}[
name = axis1,
width = 0.5\textwidth, 
height = 0.35\textwidth,
title = {Eigenvector method ($EM$)},%$\alpha=4$,
title style = {align=center, font=\small},
xmin = 0,
xmax = 3,
ymin = 0,
ymax = 1,
ymajorgrids,
xlabel = Value of $\alpha$,
xlabel style = {font=\small},
%ylabel = Probability decrease compared to $RR$,
tick label style = {/pgf/number format/fixed, font=\small},
]
%\draw (axis cs:\pgfkeysvalueof{/pgfplots/xmin},0.5)  -- (axis cs:\pgfkeysvalueof{/pgfplots/xmax},0.5);
% 2014
\addplot[blue,smooth,very thick,loosely dashdotdotted] coordinates {
(0,0)
(0.1,0.00143131514012411)
(0.2,0.00555898613049284)
(0.3,0.012082511924933)
(0.4,0.0206895476324229)
(0.5,0.0311142188128462)
(0.6,0.0431721601384325)
(0.7,0.0567698818306003)
(0.8,0.0718930160245404)
(0.9,0.0885813274297219)
(1,0.106898521512876)
(1.1,0.126903196981335)
(1.2,0.148625035730276)
(1.3,0.172048269285935)
(1.4,0.197102898568633)
(1.5,0.223663087400813)
(1.6,0.251551509835879)
(1.7,0.280548108266634)
(1.8,0.31040163463114)
(1.9,0.340842439532486)
(2,0.371595190531059)
(2.1,0.402390491345433)
(2.2,0.432974692366788)
(2.3,0.463117491430804)
(2.4,0.492617194407116)
(2.5,0.521303721520355)
(2.6,0.54903960132145)
(2.7,0.575719291581727)
(2.8,0.601267211844145)
(2.9,0.625634875447807)
(3,0.64879748014073)
};
% 2015
\addplot[black,smooth,very thick,loosely dotted] coordinates {
(0,0)
(0.1,0.00125150873618022)
(0.2,0.00540431837207313)
(0.3,0.0131500331886501)
(0.4,0.0252705661872268)
(0.5,0.0425641720093028)
(0.6,0.0657379367474185)
(0.7,0.0952804529777013)
(0.8,0.131340532487417)
(0.9,0.173643629222985)
(1,0.221472237734339)
(1.1,0.273720114411353)
(1.2,0.329009050112382)
(1.3,0.385840182932)
(1.4,0.442745979591959)
(1.5,0.498414631119146)
(1.6,0.551771181446481)
(1.7,0.602013099209039)
(1.8,0.648607748999025)
(1.9,0.69126379153801)
(2,0.729888742120566)
(2.1,0.764542536606253)
(2.2,0.79539369796548)
(2.3,0.822681682262129)
(2.4,0.846686715105777)
(2.5,0.86770698373865)
(2.6,0.886042284532189)
(2.7,0.90198293453676)
(2.8,0.915802750040962)
(2.9,0.92775503391461)
(3,0.938070705518948)
};
% 2016
\addplot[ForestGreen,smooth,very thick,loosely dashed] coordinates {
(0,0)
(0.1,0.00133569432290308)
(0.2,0.00571374579994997)
(0.3,0.0136782109603689)
(0.4,0.0257053964080553)
(0.5,0.0421368558577609)
(0.6,0.0631213068717554)
(0.7,0.0885796886194158)
(0.8,0.118202487977308)
(0.9,0.151479839491359)
(1,0.187756347584574)
(1.1,0.226297486241676)
(1.2,0.26635416436758)
(1.3,0.30721560572049)
(1.4,0.348245817024596)
(1.5,0.388903506978447)
(1.6,0.428748257883252)
(1.7,0.467436932581382)
(1.8,0.504714198982922)
(1.9,0.540400286175851)
(2,0.574378135301866)
(2.1,0.606581256247363)
(2.2,0.636982958178248)
(2.3,0.66558719407057)
(2.4,0.692421005801914)
(2.5,0.717528425455598)
(2.6,0.74096563518803)
(2.7,0.762797179176438)
(2.8,0.783093035020808)
(2.9,0.801926375448038)
(3,0.819371877065012)
};
% 2017
\addplot[red,smooth,very thick] coordinates {
(0,0)
(0.1,0.0017212992461858)
(0.2,0.00769660904351459)
(0.3,0.0192274508818564)
(0.4,0.0375799119415641)
(0.5,0.0637439771900766)
(0.6,0.0981783358266153)
(0.7,0.140625836664407)
(0.8,0.190075844553548)
(0.9,0.244897563509419)
(1,0.30310050183272)
(1.1,0.362634929070721)
(1.2,0.421647082098757)
(1.3,0.478639923504891)
(1.4,0.53253326403563)
(1.5,0.582646135087966)
(1.6,0.628634038833085)
(1.7,0.670409674393461)
(1.8,0.708066135561333)
(1.9,0.741812197572544)
(2,0.771922620311733)
(2.1,0.798702604207841)
(2.2,0.822463905429508)
(2.3,0.843509781379041)
(2.4,0.862126231468571)
(2.5,0.878577512037718)
(2.6,0.893104426171188)
(2.7,0.905924331118146)
(2.8,0.917232146307306)
(2.9,0.927201891866351)
(3,0.935988459503662)
};
% 2018
\addplot[brown,smooth,very thick,dashdotted] coordinates {
(0,0)
(0.1,0.00121227353787118)
(0.2,0.00529500150023815)
(0.3,0.012896304099627)
(0.4,0.0245614482368087)
(0.5,0.0406365448695442)
(0.6,0.061190517912153)
(0.7,0.085978668022128)
(0.8,0.114461075973981)
(0.9,0.14587270020824)
(1,0.179326584016263)
(1.1,0.213923905867572)
(1.2,0.248846969761089)
(1.3,0.283420755605456)
(1.4,0.317139956230737)
(1.5,0.349667047786797)
(1.6,0.380811059589796)
(1.7,0.410496880889599)
(1.8,0.438732810689194)
(1.9,0.465581221734945)
(2,0.491134690208097)
(2.1,0.515498152324512)
(2.2,0.538776600095879)
(2.3,0.561067350952226)
(2.4,0.582455815200415)
(2.5,0.603013772321408)
(2.6,0.622799336935248)
(2.7,0.641857980774758)
(2.8,0.66022414503984)
(2.9,0.677923115669957)
(3,0.69497294110014)
};
\end{axis}

\begin{axis}[
at = {(axis1.south east)},
xshift = 0.1\textwidth,
width = 0.5\textwidth, 
height = 0.35\textwidth,
title = {Row geometric mean method ($RGM$)},%$\alpha=4$,
title style = {align=center, font=\small},
xmin = 0,
xmax = 3,
ymin = 0,
ymax = 1,
ymajorgrids,
xlabel = Value of $\alpha$,
xlabel style = {font=\small},
%ylabel = Probability decrease compared to $RR$,
tick label style = {/pgf/number format/fixed},
legend entries={$2014 \qquad$,$2015 \qquad$,$2016 \qquad$,$2017 \qquad$,$2018$},
legend style = {at={(-0.1,-0.3)},anchor=north,legend columns = 5,font=\small}
]
%\draw (axis cs:\pgfkeysvalueof{/pgfplots/xmin},0.5)  -- (axis cs:\pgfkeysvalueof{/pgfplots/xmax},0.5);
% 2014
\addplot[blue,smooth,very thick,loosely dashdotdotted] coordinates {
(0,0)
(0.1,0.00142682327088737)
(0.2,0.00551473391664741)
(0.3,0.0119053802006071)
(0.4,0.0202061728456879)
(0.5,0.0300535275043967)
(0.6,0.0411549902371943)
(0.7,0.0533073312631785)
(0.8,0.0663945983648216)
(0.9,0.0803732688737497)
(1,0.0952515744973957)
(1.1,0.111068283303597)
(1.2,0.12787405626682)
(1.3,0.145716720026692)
(1.4,0.164630647715522)
(1.5,0.184629843286704)
(1.6,0.205704103692526)
(1.7,0.227817618663039)
(1.8,0.250909441169215)
(1.9,0.274895355389252)
(2,0.299670752704214)
(2.1,0.325114191491678)
(2.2,0.351091365462699)
(2.3,0.377459243767842)
(2.4,0.40407017969538)
(2.5,0.430775817397055)
(2.6,0.457430659576945)
(2.7,0.483895193631527)
(2.8,0.510038508351138)
(2.9,0.535740366437827)
(3,0.560892728140694)
};
% 2015
\addplot[black,smooth,very thick,loosely dotted] coordinates {
(0,0)
(0.1,0.00124595334612832)
(0.2,0.00534343273038703)
(0.3,0.0128838084493553)
(0.4,0.0244870856807816)
(0.5,0.0407319820490394)
(0.6,0.0620706832160068)
(0.7,0.0887435584317835)
(0.8,0.120713051772756)
(0.9,0.157633851510134)
(1,0.198868140800899)
(1.1,0.243543138481266)
(1.2,0.290637768268054)
(1.3,0.339079866822128)
(1.4,0.387836003761278)
(1.5,0.435981248015092)
(1.6,0.482743258603031)
(1.7,0.527521334120652)
(1.8,0.569885111330882)
(1.9,0.609559262623741)
(2,0.646400372693323)
(2.1,0.680370982046206)
(2.2,0.711514259819256)
(2.3,0.739931352019111)
(2.4,0.765762347282257)
(2.5,0.789171050109291)
(2.6,0.810333303886917)
(2.7,0.829428384843638)
(2.8,0.846632916712201)
(2.9,0.862116772712025)
(3,0.876040492896224)
};
% 2016
\addplot[ForestGreen,smooth,very thick,loosely dashed] coordinates {
(0,0)
(0.1,0.00133313639307992)
(0.2,0.00568481460714784)
(0.3,0.0135501150009253)
(0.4,0.0253278739295741)
(0.5,0.0412583486061093)
(0.6,0.061376642510811)
(0.7,0.0854943907966586)
(0.8,0.11321502544186)
(0.9,0.143978625409733)
(1,0.17712487311598)
(1.1,0.211959771901751)
(1.2,0.247813774506878)
(1.3,0.28408394292867)
(1.4,0.320258185529759)
(1.5,0.355923680880113)
(1.6,0.39076366392887)
(1.7,0.424547099366109)
(1.8,0.457115068957311)
(1.9,0.488366597003481)
(2,0.51824556031287)
(2.1,0.546729480528787)
(2.2,0.573820423297052)
(2.3,0.599537893785781)
(2.4,0.623913456758967)
(2.5,0.64698675910742)
(2.6,0.668802644767613)
(2.7,0.689409093270404)
(2.8,0.70885576369017)
(2.9,0.727192974675789)
(3,0.74447099360487)
};
% 2017
\addplot[red,smooth,very thick] coordinates {
(0,0)
(0.1,0.00171174502436212)
(0.2,0.00760802205324636)
(0.3,0.0188911443887668)
(0.4,0.0367108962778198)
(0.5,0.0619539344480255)
(0.6,0.0950264285163465)
(0.7,0.135699962560839)
(0.8,0.183080647300668)
(0.9,0.235719199931006)
(1,0.291827241348104)
(1.1,0.349531467150953)
(1.2,0.407097315333497)
(1.3,0.463079725302856)
(1.4,0.516391409221878)
(1.5,0.566303102459212)
(1.6,0.612400313866768)
(1.7,0.654520102382929)
(1.8,0.692685015697566)
(1.9,0.727044097218893)
(2,0.757825137229385)
(2.1,0.78529869743752)
(2.2,0.809752558333501)
(2.3,0.831474546838274)
(2.4,0.850741677639695)
(2.5,0.867813823013931)
(2.6,0.882930499257425)
(2.7,0.896309715680148)
(2.8,0.908148131397313)
(2.9,0.918621996977614)
(3,0.927888528847249)
};
% 2018
\addplot[brown,smooth,very thick,dashdotted] coordinates {
(0,0)
(0.1,0.00121037956672582)
(0.2,0.00527762179183037)
(0.3,0.0128301030638667)
(0.4,0.0243872998296676)
(0.5,0.0402654876972345)
(0.6,0.0605027771643151)
(0.7,0.0848259443424797)
(0.8,0.112671236382934)
(0.9,0.143255291942042)
(1,0.175677607210356)
(1.1,0.209028966495311)
(1.2,0.242482863150384)
(1.3,0.275356227436173)
(1.4,0.307136610802406)
(1.5,0.337481173439737)
(1.6,0.366196732373771)
(1.7,0.393210303177914)
(1.8,0.418537563708508)
(1.9,0.442253985847658)
(2,0.464470981985394)
(2.1,0.48531769324877)
(2.2,0.504928025597993)
(2.3,0.523432063375992)
(2.4,0.540950867204704)
(2.5,0.55759373444589)
(2.6,0.573457155499867)
(2.7,0.588624872465284)
(2.8,0.603168604919125)
(2.9,0.617149137786581)
(3,0.630617566528572)
};
\end{axis}
\end{tikzpicture}
\end{figure}

%\end{document}

Figure~\ref{Fig3} depicts the value of the inequality measure $HHI^*$ as a function of parameter $\alpha$ in the five seasons between 2014 and 2018. This can be especially relevant for a decision-maker who should fix the rules before the start of a season with having in mind a maximal level of inequality. For instance, choosing $\alpha = 1$ provides that the normalised Herfindahl--Hirschman index will not exceed $0.31$ if the given season remains more balanced than the 2017 season.

As our intuition suggests, a higher $\alpha$ results in a more unequal distribution. It is also worth noting that $HHI^*$ is consistently smaller in the case of row geometric mean than for the eigenvector method. This can be explained by an important finding of \citet[p.~28]{HermanKoczkodaj1996}: ``\emph{After all, geometric means are `means', and therefore the Euclidean metric (which `equalizes' differences of all elements) generates better results for GM than for EV solutions.}'' In other words, the eigenvector method allows for bigger differences between the ratios of the weights, which is---according to Definition~\ref{Def36}---unfavourable for $HHI$, thus for $HHI^*$, too.
Furthermore, while the order of the seasons by the normalised Herfindahl--Hirschman index for a given $\alpha$ is relatively robust, the shape of the five lines varies.

\begin{table}[ht!]
  \centering
  \caption{Revenue allocation based on performance in the 2018 season}
  \label{Table6}
\rowcolors{3}{gray!20}{}
\begin{threeparttable}
    \begin{tabularx}{0.9\textwidth}{Lccccc} \toprule \hiderowcolors
    Constructor & Position & Column 2 revenue & Share (\%) & \multicolumn{2}{c}{Indifferent $\alpha$} \\ \showrowcolors   
          &       & Million US dollars      &       & $EM$ & $RGM$ \\ \bottomrule
    Mercedes & 1     & 66    & 18.86 & 0.35  & 0.35 \\
    Ferrari & 2     & 56    & 16.00 & 0.35  & 0.35 \\
    Red Bull & 3     & 46    & 13.14 & ---   & --- \\
    Renault & 4     & 38    & 10.86 & ---   & --- \\
    Haas  & 5     & 35    & 10.00 & 0     & 0 \\
    McLaren & 6     & 32    & 9.14  & 0.13  & 0.14 \\
    Racing Point & 7     & 24    & 6.86  & 0.6   & 0.61 \\
    Sauber & 8     & 21    & 6.00  & 0.55  & 0.55 \\
    Toro Rosso & 9     & 17    & 4.86  & 0.59  & 0.59 \\
    Williams & 10    & 15    & 4.29  & 0.6   & 0.6 \\ \bottomrule
    \end{tabularx}
\begin{tablenotes} \footnotesize
\item Column 2 revenue is paid on a sliding scale depending on the teams' finishing position. Source: \url{https://www.racefans.net/2019/03/03/formula-1-teams-prize-money-payments-for-2019-revealed/}
\end{tablenotes}
\end{threeparttable}
\end{table}

Finally, Table~\ref{Table6} summarises the allocation of revenues in 2019, based on the 2018 season.
If the pot would be shared equally, then each team would receive $10$ million US dollars. Hence, the constructor Haas is indifferent between its actual position and our proposal with $\alpha = 0$.
By increasing the parameter, only the three top teams (Mercedes, Ferrari, Red Bull) will receive more money, implying that Renault could not receive its actual share (10.86\%) under any $\alpha$. The same holds for Red Bull.\footnote{~Figure~\ref{Fig1b} illustrates that a team might receive its maximal share at a particular $\alpha > 0$.}
Mercedes and Ferrari would be indifferent if $\alpha = 0.35$ as the last two columns of Table~\ref{Table5} show. McLaren prefers our proposal until $\alpha$ stands at a relatively small level. The remaining four teams support the current allocation only if $\alpha$ exceeds $0.55$, when the normalised Herfindahl--Hirschman index becomes about $0.05$ (see Figure~\ref{Fig3}).
Consequently, the pairwise comparison approach contains some non-linearity: the revenue share of certain teams is a non-monotonic function of the inequality parameter, thus the teams cannot be separated into two sets such that one group prefers a small $\alpha$, and the other benefits from a large $\alpha$.

\section{Conclusions} \label{Sec5}

We have presented a model to share the revenue of an industry through the example of the Formula One World Constructors' Championship. The methodology is based on multiplicative pairwise comparison matrices and makes possible to tune the inequality of the allocation by its single parameter. Since the choice of the weighting method has only a marginal effect in this particular application, we recommend using the row geometric mean, which has favourable theoretical properties.
The proposed technique has an important advantage over the official points scoring system of Formula One as it is independent of the somewhat arbitrary valuation given to the race prizes.

Besides offering a novel way for sharing Formula One prize money among the constructors, our methodology can be applied in any area where resources/revenues should be allocated among groups whose members are ranked several times.
Potential examples include further racing competitions such as \href{https://en.wikipedia.org/wiki/Grand_Prix_motorcycle_racing}{Grand Prix motorcycle racing}, combined events in athletics like \href{https://en.wikipedia.org/wiki/Decathlon}{decathlon} and \href{https://en.wikipedia.org/wiki/Heptathlon}{heptathlon}, the performance of countries in the \href{https://en.wikipedia.org/wiki/Olympic_Games}{Olympic Games}, schools on the basis of student tests in various subject areas, or even workplaces where individual contributions on various projects are ranked. The suggested approach is also able to rank Formula One drivers. While this would probably lead to an \emph{incomplete pairwise comparison matrix} as some pilots cannot be compared, both the eigenvector and the row geometric mean method have been generalised for such matrices \citep{Kwiesielewicz1996, BozokiFulopRonyai2010}.

Our paper can inspire further research in various fields.
The first possible direction is the analysis of other weighting methods concerning scale invariance.
The proposed methodology provides a new ranking of the Formula One constructors, which can be compared to alternative ranking systems.
Finally, representing race results in a pairwise comparison matrix may encourage novel ways to measure competitive balance by inconsistency indices \citep{Brunelli2018}.

\section*{Acknowledgements}
\addcontentsline{toc}{section}{Acknowledgements}
\noindent
We are grateful to \emph{S\'andor Boz\'oki} and \emph{Tam\'as Halm} for useful comments. \\
Four anonymous reviewers provided valuable remarks and suggestions on earlier drafts. \\
%We are indebted to the \href{https://en.wikipedia.org/wiki/Wikipedia_community}{Wikipedia community} for contributing to our research by providing the data used. \\
The research was supported by the MTA Premium Postdoctoral Research Program grant
PPD2019-9/2019, the NKFIH grant K 128573, and the \'UNKP-19-3-III-BCE-97 New
National Excellence Program of the Ministry for Innovation and Technology.

\bibliographystyle{apalike}
\bibliography{All_references}

\end{document}